%% file: gu.tex
\documentclass[proceedings]{stacs}
\stacsheading{2010}{429-440}{Nancy, France}
\firstpageno{429}

\include{style}

\newcommand{\co}{\mathrm{co}}
\newcommand{\EEE}{\mathrm{EEE}}
\newcommand{\NEEE}{\mathrm{NEEE}}
\newcommand{\coNEEE}{\co\NEEE}
\newcommand{\CoNEEE}{\coNEEE}

\newcommand{\strings}{\{0,1\}^*}
\newcommand{\myset}[2]{ \left\{ #1 \;\left|\; #2 \right. \right\} }
\newcommand{\mysetl}[2]{ \left\{\left. #1 \right| #2 \right\} }
\newcommand{\QP}{{\rm QP}}

\newcommand{\sA}{\mathcal{A}}

\usepackage{enumerate}

\newcommand{\twonO}{2^{n^{\Omega(1)}}}

\renewcommand{\poly}{\mathrm{poly}}

\begin{document}

\title[Collapsing and Separating Completeness Notions]{Collapsing and
  Separating Completeness Notions under Average-Case and Worst-Case Hypotheses }

\author[linkedin]{X. Gu}{Xiaoyang Gu}
\thanks{Gu's research
was supported in part by NSF grants 0652569 and 0728806.}
\address[linkedin]{LinkedIn Corporation}

\author[uw]{J. M. Hitchcock}{John M. Hitchcock}
\address[uw]{Department of Computer Science, University of Wyoming}
\thanks{Hitchcock's research was supported in part by NSF grants 0515313 and
   0652601 and by an NWO travel grant.  Part of this research was done
   while this author was on sabbatical at CWI}

\author[isu]{A. Pavan}{A. Pavan}
\address[isu]{Department of Computer Science, Iowa State University}
\thanks{Pavan's research was supported in part by NSF grants 0830479 and 0916797.}

\input{abstract}
\maketitle
\input{intro}

\input{section2}

\input{li}

\input{tm}

\input{gu.bbl}
\end{document}

%% file: style.tex
\usepackage{amsmath}
\usepackage{amsthm}
\usepackage{amssymb}
\usepackage{epsfig}

\newcommand{\N}{\mathbb{N}}

\newcommand{\ignore}[1]{}

\newcommand{\newfontobj}[2]{
  \newcommand{#1}[1]{
    \expandafter\def\csname##1\endcsname{{#2 ##1}}}}

\newfontobj{\class}{\rm}
\newfontobj{\lang}{\bf}
\newfontobj{\oper}{\rm}

\newcommand{\poly}{\mbox{\tt poly}}

\newcommand{\maj}{\mbox{\tt MAJ}}

\class{PSPACE}          
\class{AM}              
\class{AMTIME}
\class{AMSUBEXP}
\class{MA}
\class{LINSPACE}
\class{NP}
\class{NLIN}
\class{P}
\class{RP}
\class{BPP}
\class{advBPP}
\class{DTIME}
\class{ZPP}
\class{EXPSPACE}
\class{PH}
\class{IP}
\class{AVP}

\class{AvgTIME} \class{DistNP} \class{AZPP} \class{AVE}
\class{DistE} \class{DistEXP} \class{E} \class{NE} \class{UE}
\class{EXP} \class{pcomp} \class{NEXP} \class{UP} \class{DTIME}
\class{NTIME} \class{R} \class{NSUBEXP} \class{SUBEXP} \class{PF}
\class{LINcomp}
\class{BPTIME}
\class{DistNLIN}
\class{AME}
\class{SIZE}
\class{NSIZE}

\newcommand{\coNP}{{\mbox {\rm co-NP}}}
\newcommand{\CoNP}{{\mbox {\rm co-NP}}}

\newcommand{\SAT}{\mathrm{SAT}}

\newcommand{\ppoly}{\mbox{\rm P/poly}}

\newcommand{\io}[1]{\mbox{\rm\tiny io}{#1}}

\newcommand{\pair}[1]{\mbox{$\langle #1 \rangle$}}

\newcommand{\lemmahack}{null}

\newcommand{\theoremhack}{null}

\newcommand{\claimhack}{null}

%% file: abstract.tex
\begin{abstract}

This paper presents the following results on sets that are complete for $\NP$.
\begin{enumerate}[\upshape (i)]

\item \label{abs:NP} If there is a problem in $\NP$ that requires
  $\twonO$ time at almost all lengths, then every many-one NP-complete
  set is complete under length-increasing reductions that are computed
  by polynomial-size circuits.

\item If there is a problem in $\CoNP$ that cannot be solved by
  polynomial-size nondeterministic circuits, then every many-one
  complete set is complete under length-increasing reductions that are
  computed by polynomial-size circuits.

\item If there exist a one-way permutation that is secure against
  subexponential-size circuits and there is a hard tally language in
  $\NP \cap \CoNP$, then there is a Turing complete language for $\NP$
  that is not many-one complete.

\end{enumerate}

\noindent
Our first two results use worst-case hardness hypotheses whereas
earlier work that showed similar results relied on average-case or
almost-everywhere hardness assumptions.  The use of average-case and
worst-case hypotheses in the last result is unique as previous results
obtaining the same consequence relied on almost-everywhere hardness
results.


\end{abstract}

%% file: intro.tex
\newcommand{\EEElog}{\EEE/\log}

\section{Introduction}

It is widely believed that many important problems in $\NP$ such as
satisfiability, clique, and discrete logarithm are
exponentially hard to solve.  Existence of such intractable problems has a
bright side: research has
shown that we can use this kind of intractability to our advantage to
gain a better understanding of computational complexity, for
derandomizing probabilistic computations, and for designing
computationally-secure cryptographic primitives.  For example, if
there is a problem in $\EXP$ (such as any of the aforementioned
problems) that has $\twonO$-size worst-case circuit complexity (i.e.,
that for all sufficiently large $n$, no subexponential size circuit
solves the problem correctly on all instances of size $n$), then it
can be used to construct pseudorandom generators. Using these
pseudorandom generators,  $\BPP$ problems can be solved
in deterministic quasipolynomial time~\cite{ImpWig97}.
Similar average-case hardness assumptions on
the discrete logarithm and factoring problems have important
ramifications in cryptography.
While these hardness assumptions have
been widely used in cryptography and derandomization, more recently
Agrawal \cite{Agr02} and Agrawal and Watanabe~\cite{AgrWat09} 
showed that they are also useful for
improving our understanding of $\NP$-completeness.  
In this paper, we provide further applications of such hardness assumptions.

\subsection{Length-Increasing Reductions}

A language is $\NP$-complete if every language in $\NP$ is {\em
reducible} to it. While there are several ways to define the notion of
reduction, the most common definition uses 
polynomial-time computable many-one functions. Many natural problems that
arise in practice have been shown to be NP-complete using polynomial-time
computable many-one reductions.
However, it has been observed that all known $\NP$-completeness results 
hold when we restrict the notion of reduction.
For example, $\SAT$ is complete under
polynomial-time reductions that are one-to-one and length-increasing.
In fact, all known many-one complete problems for $\NP$ are complete
under this type of reduction \cite{BerHar77}.  This raises the
following question: are there languages that are complete under
polynomial-time many-one reductions but not complete under
polynomial-time, one-to-one, length-increasing reductions?
Berman~\cite{Berm77} showed that every many-one complete set for $\E$
is complete under one-to-one, length-increasing reductions. Thus for
$\E$, these two completeness notions coincide.  A weaker result is
known for $\NE$. Ganesan and Homer~\cite{GanHom92} showed that all
$\NE$-complete sets are complete via one-to-one reductions that are
exponentially honest.

For NP, until recently there had not been any progress on this
question.  Agrawal~\cite{Agr02} showed that if one-way
permutations exist, then all NP-complete sets are complete via
one-to-one, length-increasing reductions that are computable by
polynomial-size circuits.  Hitchcock and
Pavan \cite{Hitchcock:CRNPCS} showed that
$\NP$-complete sets are complete under length-increasing $\ppoly$
reductions under the
measure hypothesis on $\NP$ \cite{Lutz:MSDHL}.  Recently Buhrman et
al.~improved the latter result to show that if the measure hypothesis
holds, then all NP-complete sets are complete via length-increasing,
$\P/$-computable  functions with $\log \log n$ bits of advice~\cite{BuhrmanLatest}.  
More recently, Agrawal and Watanabe~\cite{AgrWat09} showed that if there
exist regular one-way functions, then all NP-complete sets are
complete via one-one, length-increasing, $\ppoly$-computable
reductions.  All the hypotheses used in these works require the
existence of an {\em almost-everywhere hard} language or an {\em
  average-case hard} language in $\NP$.

In the first part of this paper, we consider hypotheses that only
concern the {\em worst-case hardness} of languages in $\NP$. Our first
hypothesis concerns 
the deterministic time complexity 
of languages in
$\NP$.  We show that if there is a language in $\NP$ for which 
 every correct algorithm spends more than
$2^{n^{\epsilon}}$ time at almost all lengths, then NP-complete languages
are complete via $\ppoly$-computable, length-increasing reductions.
The second hypothesis concerns nondeterministic circuit complexity of
languages in $\CoNP$. We show that if there is a language in $\CoNP$
that cannot be solved by nondeterministic polynomial-size circuits,
then all NP-complete sets are complete via length-increasing
$\ppoly$-computable reductions.  For more formal statements of the
hypotheses, we refer the reader to Section~\ref{li}.  We stress
that these hypotheses require only worst-case hardness. The worst-case
hardness is of course required at every length, a technical condition
that is necessary in order to build a reduction that works at every
length rather than just infinitely often. 

\subsection{Turing Reductions versus Many-One Reductions}

In the second part of the paper we study the completeness notion obtained by
allowing a more general notion of reduction---Turing reduction.
Informally, with Turing reductions an instance of a problem can be solved by
asking polynomially many (adaptive) queries about the instances of the other problem.
A language in $\NP$ is Turing complete if there is a polynomial-time Turing
reduction to it from every other language in $\NP$. Though many-one
completeness is the most commonly used completeness notion, Turing
completeness also plays an important role in complexity theory.
Several properties of Turing complete sets are closely tied to the
separation of complexity classes. For example, Turing complete sets for
EXP are sparse if and only if EXP contains polynomial-size circuits.
Moreover, to capture our intuition
that a complete problem is easy, then the entire class is easy, Turing
reductions seem to be the  ``correct'' reductions to define completeness.
In fact, the seminal paper of Cook~\cite{Cook71} used Turing reductions to
define completeness, though Levin~\cite{Levin73} used many-one reductions.

This raises the question of whether there is a Turing complete language for
$\NP$ that is not many-one complete. Ladner, Lynch and
Selman~\cite{LadnerLynchSelman75} posed this question in 1975, thus making it
one of the oldest problems in complexity theory. This question is 
completely resolved for exponential time classes such as $\EXP$ and
$\NEXP$~\cite{Wata87a,BuHoTo91}. We know that for both these
classes many-one completeness differs from Turing-completeness.
However progress on the $\NP$ side has been very slow. Lutz and
Mayordomo~\cite{Lutz:CVKL} were the first to provide evidence that
Turing completeness differs from many-one completeness. They showed that if
the measure hypothesis holds, then the completeness notions differ.
Since then a few other weaker hypotheses 
have been used to achieve the separation of Turing completeness
from many-one completeness~\cite{AmbBen00,PavSel02b,PavanSelman04,Hitchcock:PBSDNPC,Pavan03}.

All the hypotheses used in the above works are considered ``strong''
hypotheses as they require the existence of
an {\em almost everywhere hard} language in $\NP$. That is, there is  a
language $L$ in $\NP$ and every algorithm that decides $L$
takes  exponential-time an {\em all but finitely many} strings. A drawback of
these hypotheses is that we do not have any candidate languages in $\NP$ that
are believed to be almost everywhere hard.

It has been open whether we can achieve the separation using more believable
hypotheses that involve average-case hardness or worst-case hardness. None of
the proof techniques used earlier seem to achieve this, as the they
crucially depend on the almost everywhere hardness. 

In this paper, for the first time, we achieve the separation between Turing
completeness and many-one completeness using average-case and worst-case
hardness hypotheses.  
We consider two hypotheses.  The first hypothesis states that there
exist $2^{n^{\epsilon}}$-secure one-way permutations and the second
hypothesis states that there is a language in $\NEEE \cap \CoNEEE$
that can not be solved in triple exponential time with logarithmic advice,
i.e, $\NEEE \cap \CoNEEE \not\subseteq \EEE/\log$.
We show that if both of these
hypothesis are true, then there is a Turing complete language in $\NP$
that is not many-one complete.

The first hypothesis is an average-case hardness hypothesis and has been 
studied extensively in past.
The second hypothesis is a worst-case hardness hypothesis. At first
glance, this hypothesis may look a little esoteric, however, it is
only used to obtain hard tally languages in $\NP \cap \CoNP$ that are
sufficiently sparse. 
Similar hypotheses involving double
and triple exponential-time classes have been used earlier in
the literature~\cite{BeigelFeigenbaum92,FFLS94,HNOS96,FFLN98}.

We use length-increasing reductions as a tool to achieve the separation of
Turing completeness from many-one completeness.
We first show that if one-way
permutations exist then $\NP$-complete sets are complete via
length-increasing, quasipolynomial-time computable reductions.
We then show that if the second hypothesis holds,
then there is a Turing complete language for $\NP$ that is not complete
via quasi polynomial-time, length-increasing reductions. Combining these
two results we obtain our separation result.

%% file: section2.tex
\section{Preliminaries}\label{prelims}
In the paper, we use the binary alphabet $\Sigma = \{0,1\}$.  Given a
language $A$, $A_n$ denotes the characteristic sequence of $A$ at
length $n$.  We also view $A_n$ as a boolean function from $\Sigma^n$ to
$\Sigma$.
For languages $A$ and $B$, we say that $A = \io B$, if
$A_n = B_n$ for infinitely many $n$.  For a complexity class
$\mathcal{C}$,
we say that $A \in \io{\mathcal{C}}$ if there is a language $B \in
{\mathcal{C}}$ such that $A = \io B$. 

For a boolean function $f: \Sigma^n \rightarrow \Sigma$, $CC(f)$ is the
smallest number $s$ such that there is circuit of size $s$
that computes $f$. 
A function $f$ is quasipolynomial time computable ($\QP$-computable)
if can be computed deterministically in time $O(2^{\log^{O(1)} n})$.
We will use the triple exponential time class $\EEE =
\DTIME(2^{2^{2^{O(n)}}})$, and its nondeterministic counterpart $\NEEE$.

A language $L$ is in $\NP/\poly$ if there is a polynomial-size circuit $C$ and
a polynomial $p$ such that for every $x$, $x$ is in $L$ if and only if there
is a $y$ of length $p(|x|)$ such that $C(x, y) = 1$.

Our proofs make use a variety of results from  approximable sets, instance
compression, derandomization and hardness amplification. We mention the
results that we need.

\begin{definition} 
A language $A$ is {\em $t(n)$-time 2-approximable} \cite{Beigel87a} if there
is a function $f$ computable in time $t(n)$ such that for all strings
$x$ and $y$, $f(x, y) \neq A(x)A(y)$.

A language $A$ is {\em io-lengthwise t(n)-time
2-approximable} if there is a function $f$ computable in time
$t(n)$ such that for infinitely many $n$, for every pair of $n$-bit
strings $x$ and $y$, $f(x, y) \neq A(x)A(y)$.
\end{definition}
\noindent
Amir, Beigel, Gasarch~\cite{AmBeGa03} proved that every
polynomial-time 2-approximable set is in $\P/\poly$.  Their proof also
implies the following extension for a superpolynomial function $t(n)$.
\begin{theorem}[\cite{AmBeGa03}]\label{2mctheorem}
If $A$ is io-lengthwise $t(n)$-time 2-approximable,
then for infinitely many $n$, $CC(A_n) \leq t^2(n)$.
\end{theorem}

Given a language $H'$ in $\CoNP$, let $H$ be $\{\langle x_1, \cdots,
x_n\rangle~|~ |x_1| = \cdots = |x_n| = n, x_i \in
H'\}$.
Observe that a $n$-tuple consisting of strings of length $n$ can be
encoded by a string of length $n^2$. From now we view a string of length
$n^2$ as an $n$-tuple of strings of length $n$.

\begin{theorem}[\cite{FortnowSanthanam08,BuhrmanHitchcock08}]\label{fs-bh}
Let $H$ and $H'$ be defined as above. Suppose there is a language $L$,
a polynomial-size circuit family $\{C_m\}$, and a polynomial $p$ such that for
infinitely many $n$, for every $x \in \Sigma^{n^2}$, 
$x$ is in $H$ if and only if there is a string $y$ of length $p(n)$ such
that $C(x, y)$ is in $L^{\leq n}$. Then $H'$ is in $\io\NP/poly$.
\end{theorem}

The proof of Theorem \ref{fs-bh} is similar to the proofs in
\cite{FortnowSanthanam08,BuhrmanHitchcock08}.  The difference is
rather than having a polynomial-time many-one reduction, here we have
a $\NP/\poly$ many-one reduction which works infinitely often.  The
nondeterminism and advice in the reduction can be absorbed into the
final $\NP/\poly$ decision algorithm.  The $\NP/\poly$ decision
algorithm works infinitely often, corresponding to when the
$\NP/\poly$ reduction works.

\begin{definition}
A function $f:\{0,1\}^n\rightarrow \{0,1\}^m$ is $s$-{\em secure} if for every $\delta<1$,
every $t\leq \delta s$, and every circuit $C:\{0,1\}^n\rightarrow
\{0,1\}^m$ of size $t$, $\Pr[C(x) = f(x)]\leq 2^{-m}+ \delta$.
A function $f:\strings\rightarrow\strings$ is $s(n)$-{\em secure} if it is $s(n)$-secure
at all but finitely many length $n$.
\end{definition}

\begin{definition}
An {\em $s(n)$-secure one-way permutation} is a polynomial-time
computable bijection $\pi:\strings\rightarrow\strings$ such that
$|\pi(x)| = |x|$ for all $x$ and $\pi^{-1}$ is $s(n)$-secure.
\end{definition}
\noindent
Under widely believed average-case hardness assumptions about the
hardness of the RSA cryptosystem or the discrete logarithm problem,
there is a secure one-way permutation~\cite{GoldreichLevinNisan95}.

\begin{definition}
A {\em pseudorandom generator (PRG) family} is a collection of
functions $G=\{G_n:\{0,1\}^{m(n)} \rightarrow \{0,1\}^n\}$ such that
$G_n$ is uniformly computable in time $2^{O(m(n))}$ and for every
circuit of $C$ of size $n$,
\[\left|\Pr_{x \in \{0,1\}^n}[C(x)=1] - \Pr_{y \in \{0, 1\}^{m(n)}}[C(G_n(y))
= 1\right | \leq \frac{1}{n}.\]
\end{definition}

There are many results that show that the existence of hard functions in
exponential time implies PRGs exist. We will use the following.

\begin{theorem}[\cite{NisWig94,ImpWig97}]\label{lm:pseudorandom}
If there is a language $A$ in $\E$ such that $CC(A_n) \geq
2^{n^{\epsilon}}$ for all sufficiently large $n$, then
there exist a constant $k$ and  a PRG family
$G = \{G_n:\{0,1\}^{\log^k n} \rightarrow \{0,1\}^n\}$.
\end{theorem}

%% file: li.tex
\section{Length-Increasing Reductions}\label{li}

In this section we provide evidence that many-one complete sets for NP
are complete via length-increasing reductions. We use the following hypotheses.

\noindent {\bf Hypothesis 1.} There is a language $L$ in $\NP$ and a constant
$\epsilon > 0$ such that $L$ is not in $\io\DTIME(2^{n^{\epsilon}})$.

Informally, this means that every algorithm that decides $L$ takes more than
$2^{n^{\epsilon}}$-time on at least one string at every length.

\noindent {\bf Hypothesis 2.} There is a language $L$ in $\coNP$ such that $L$
is not in $\io\NP/\poly$.

This means that every nondeterministic polynomial size circuit family that
attempts to solve $L$ is wrong on on at least one string at each length.

We will first consider the following variant of Hypothesis 1.

\noindent {\bf Hypothesis 3.} There is a language $L$ in $\NP$ and a constant
$\epsilon > 0$ such that for all but finitely many $n$, $CC(L_n) >
2^{n^{\epsilon}}$.

We will first show that Hypothesis $3$ holds, then
$\NP$-complete sets
are complete via length-increasing reductions. 
Then we describe how to modify the
proof to derive the same consequence under Hypothesis 1. We do this because
the proof is much cleaner with Hypothesis $3$. To use Hypothesis $1$ we have
to fix encodings of boolean formulas with certain properties.

\subsection{If $\NP$ has Subexponentially Hard Languages}

\begin{theorem}\label{NPthm}
If there is a language $L$ in $\NP$ and an $\epsilon > 0$ such that for
all but finitely many $n$, $CC(L_n) > 2^{n^{\epsilon}}$, 
then all $\NP$-complete sets are complete via
length-increasing, $\ppoly$ reductions.
\end{theorem}

\begin{proof}
Let $A$ be a $\NP$-complete set that is decidable in time $2^{n^k}$. Let
$L$ be a language in $\NP$ that requires $2^{n^{\epsilon}}$-size
circuits at every length.
Since $\SAT$ is complete via
polynomial-time, length-increasing reductions, it suffices to exhibit a
length-increasing, $\P/\poly$-reduction from $\SAT$ to $A$.

Let $\delta = \frac{\epsilon}{2k}$. Consider the following intermediate language
\[S = \myset{\langle x, y, z\rangle}{|x| = |z|, |y| = |x|^{\delta},
\maj[L(x), \SAT(y), L(z)] = 1}.\]

Clearly $S$ is in $\NP$. Since $A$ is $\NP$-complete, there is a many-one
reduction $f$ from $S$ to $A$. We will first show that at every length
$n$ there exist strings on which the reduction $f$ must be honest. Let

\[T_n = \myset{\langle x, z\rangle \in \{0,1\}^n \times \{0,1\}^n}{L(x) \neq
L(z), ~ \forall y\in \{0,1\}^{n^{\delta}} ~|f(\langle x, y, z\rangle)| >
n^{\delta}}\]
\begin{lemma}\label{NPlemma}
For all but finitely many $n$, $T_n \neq \varnothing$.
\end{lemma}

Assuming that the above lemma holds, we complete the proof of the
theorem. Given a length $m$, let $n = m^{1/\delta}$. Let $\langle x_n,
z_n\rangle$ be the first tuple from $T_n$.
Consider the following reduction from $\SAT$ to $A$: Given a string
$y$ of length $m$, the reduction outputs $f(\langle x_n, y,
z_n\rangle)$. Given $x_n$ and $y_n$ as advice, this reduction can be
computed in polynomial time. Since $n$ is polynomial in $m$, this is a
$\ppoly$ reduction.

By the definition of $T_n$,
$L(x_n) \neq L(z_n)$. Thus $y \in \SAT$ if and only if $\langle x_n, y,
z_n\rangle \in S$, and so $y$ is in $\SAT$ if and only if $f(\langle
x_n, y, z_n\rangle)$ is in $A$. Again, by the definition of $T_n$,
for every $y$ of length $m$,
the length of $f(\langle x_n, y, z_n\rangle)$ is bigger than $n^{\delta} =
m$.  Thus there is a $\ppoly$-computable, length-increasing reduction
from $\SAT$ to $A$.
This, together with the proof of Lemma \ref{NPlemma} we provide next,
complete the proof of Theorem \ref{NPthm}.
\end{proof}

\begin{proof}[Proof of Lemma \ref{NPlemma}]
Suppose $T_n=\varnothing$ for infinitely many $n$. We will show that this
yields a length-wise 2-approximable algorithm for $L$ at infinitely many
lengths. This enables us to contradict the hardness of $L$.
Consider the following algorithm:

\begin{enumerate}
\item Input $x$, $z$ with $|x| = |z| = n$.
\item Find a $y$ of length $n^{\delta}$ such that $|f(\langle x, y,
x\rangle)| \leq n^{\delta}$.\label{two}
\item If no such $y$ is found, Output $10$.
\item If $y$ is found, then solve the membership of $f(\langle x, y,
z\rangle)$ in $A$. If $f(\langle x, y, z\rangle) \in A$, then output
$00$, else output $11$.\label{four}
\end{enumerate}

We first bound the running time of the algorithm.
Step \ref{two} takes $O(2^{n^{\delta}})$ time. In Step \ref{four}, we
decide the membership of $f(\langle x, y, z\rangle)$ in $A$. This step
is reached only if the length of $f(\langle x, y, z\rangle)$ is at most
$n^{\delta}$. Thus the time taken to for this step is
$(2^{n^{\delta}})^k \leq 2^{n^{\epsilon/2}}$ time. Thus the total time taken
by the algorithm is bounded by $2^{n^{\epsilon}/2}$.

Consider a length $n$ at which $T_n=\varnothing$.
Let $x$ and $z$ be any strings at this length. Suppose for every
$y$ of length $n^{\delta}$, the length of $f(\langle x, y, z\rangle)$ is
at least $n^{\delta}$. Then it must be the case that $L(x) = L(z)$,
otherwise the tuple $\langle x, z\rangle$ belongs to $T_n$. Thus if the
above algorithm fails to find $y$ in Step~\ref{two}, then $L(x)L(z) \neq
10$.

Suppose the algorithm succeeds in finding a $y$ in Step~\ref{two}.
If $f(\langle x, y, z \rangle) \in A$, then at least one of
$x$ or $z$ must belong to $L$. Thus $L(x)L(z) \neq 00$. Similarly,
if $f(\langle x, y, z\rangle) \notin A$, then at least one of $x$ or $z$
does not belong to $L$, and so $L(x)L(z) \neq 11$.

Thus $L$ is
2-approximable at length $n$. If there exist infinitely many
lengths $n$, at which $T_n$ is empty, then $L$ is infinitely-often,
length-wise, $2^{n^{\epsilon}/2}$-time approximable.
By Theorem~\ref{2mctheorem}, $L$ has circuits of size $2^{n^{\epsilon}}$ at
infinitely many lengths.
\end{proof}

Now we will describe how to modify the proof if we assume that Hypothesis
1 holds. Let $L$ be the hard  language guaranteed by the hypothesis.
We will work with 3-$\SAT$. Fix an encoding of 3CNF formulas such that
formulas with same numbers of variables can be encoded as strings of same
length. Moreover, we require that the formulas $\phi(x_1, \cdots, x_n)$ and
$\phi(b_1, \cdots, b_i, x_{i+1}, \cdots, x_n)$ can be encoded as strings of
same length, where $b_i \in \{0, 1\}$. Fix a reduction $f$ from $L$ to
3-$\SAT$ such that all strings of length $n$ are mapped to formulas with
$n^r$ variables, $r\geq 1$. Let 
$3\mbox{-}\SAT' = 3\mbox{-}\SAT \cap \cup_r \Sigma^{n^r}$.  It follows that
that if there is an algorithm that decides 3-$\SAT'$ such that for infinitely
many $n$ the algorithm runs in
$2^{n^{\epsilon}}$ time on all formulas with $n^r$ variables, then $L$
is in $\io\DTIME(2^{n^{\epsilon}})$.

Now the proof proceeds exactly same as before except that we use 3-$\SAT'$
instead of $L$, i.e, our intermediate language will be
\[\{\langle x, y, z\rangle~|~\maj[3\mbox{-}\SAT'(x), \SAT(y),
3\mbox{-}\SAT'(z)]\} = 1.\]  
Consider the set $T_n$ as before. It follows that
if $T_n$ is empty at infinitely many lengths, then for infinitely many $n$, 
3-$\SAT'$ is 2-approximable on formulas with $n^r$ variables. Now we can use
the disjunctive self-reducibility of 3-$\SAT'$ to show that there is a
an algorithm that solves 3-$\SAT'$ and for infinitely many $n$, this algorithm
runs in $\DTIME(2^{n^{\epsilon}})$-time on formulas with $n^r$ variables.
This contradicts the hardness of $L$. This gives the following theorem.

\begin{theorem}
If there is a language in $\NP$ that is not in $\io\DTIME(2^{n^{\epsilon}})$,
then all $\NP$-complete sets are complete via length-increasing $\ppoly$
reductions.
\end{theorem}

\subsection{If $\CoNP$ is Hard for Nondeterministic Circuits}

In this subsection we show that  
Hypothesis 2 also implies that all NP-complete sets are complete via
length-increasing reductions.

\begin{theorem}
If there is a language $L$ in $\CoNP$ that is not in $\io\NP/poly$, then
$\NP$-complete sets are complete via $\ppoly$-computable, length-increasing 
reductions.
\end{theorem}

\begin{proof}

We find it convenient to work with $\CoNP$ rather than $\NP$. 
We will show that all $\CoNP$-complete languages are complete via
$\ppoly$, length-increasing reductions.

Let $H'$ be a
language in $\CoNP$ that is not in $\io\NP/\poly$.  Let $H$ be
\[\{\langle x_1, \cdots, x_n\rangle~|~ \forall 1 \leq i \leq n, [x_i \in H'
\mbox{ and } |x_i| = n]\}.\]
Note that every $n$-tuple that may potentially belong to $H$ can be
encoded by a string of length $n^2$.

Let $S = 0H' \cup 1\overline{SAT}$.  It is easy to show that 
$S$ is in $\CoNP$ and $S$ is not in $\io\NP/poly$. Observe that $S$ is
$\coNP$-complete via length-increasing reductions. 
Let $A$ be any $\CoNP$-complete language. It suffices to exhibit a
length-increasing reduction from $S$ to $A$.

Consider the following
intermediate language:
\[L = \{\langle x, y, z\rangle~|~|x| = |z| = |y|^2, \maj[x \in H, y \in S, z \in H] = 1\}.\]

Clearly the above language is in $\CoNP$. Let $f$ be a many-one
reduction from $L$ to $\overline{A}$. As before we will first show at
every length $n$ that there exits strings $x$ and $z$ such that for
every $y$ in $S$ the length of $f(\langle x, y, z\rangle)$ is at least
$n$.

\begin{lemma}
For all but finitely many $n$, there exist two strings $x_n$ and $z_n$ of length
$n^2$  with $H(x_n) \neq H(z_n)$ and for every $y \in S^n$, 
$|f(\langle x_n, y, z_n\rangle)| > n$.
\end{lemma}

\begin{proof}
Suppose not. Then there exist infinitely many lengths $n$ at which for every
pair of strings (of length $n^2$) $x$ and $z$ with $H(x) \neq H(z)$, there
exist a $y$ of length $n$  such that $|f(x, y, z)| \leq n$.

From this we obtain a $\NP/\poly$-reduction from $H$ to $A$ such that for infinitely
many $n$, for every $x$ of length $n^2$, $|f(x)| \leq n$. By
Theorem~\ref{fs-bh}, this implies that $H'$ is in $\io\NP/\poly$.
We now describe the reduction. Given $n$ let $z_n$ be a string (of length
$n^2$) that is not in $H$.

\begin{enumerate}
\item Input $x, |x| = n^2$. Advice: $z_n$.
\item Guess a string $y$ of length $n$.
\item If $|f(\langle x, y, z_n\rangle)| > n$, the output $\bot$.
\item Output $f(\langle x, y, z_n)$.
\end{enumerate}

Suppose $x \in H$. Since $z_n \notin H$, there exists a string 
$y$ of length $n$ such that $y \in S$ and  $|f(\langle x, y, z_n\rangle)| \leq n$. 
Consider a path that correctly guesses such a $y$.
Since $z_n \notin H$, and $y \in S$, $\langle x, y, z_n\rangle \in L$. 
Thus $f(\langle x, y, z_n\rangle) \in A^{\leq n}$.
Thus there exists at least
one path on which the reduction outputs a string from $L \cap \Sigma^{\leq
n}$. Now consider the case $x \notin H$. On any path, the reduction either
outputs $\bot$ or outputs $f(\langle x, y, z_n\rangle)$. 
Since both $z_n$ and $x$ are not in $H$, $\langle x, y, z\rangle \notin L$.
Thus $f(\langle x, y, z_n\rangle) \notin A$ for any $y$.

Thus there is a $\NP/\poly$ many-one reduction from $H$ to $L$ such that for
infinitely many $n$, the output of the reduction, on strings of length $n^2$,
on any path is at most $n$.  
By Theorem~\ref{fs-bh}, this places $H'$ in $\io\NP/\poly$.

Thus for all but finitely many lengths $n$, there exist strings $x_n$ and
$z_n$ of length $n^2$ with $H(x_n) \neq H(z_n)$ and for every $y \in S^{n}$, the
length of $f(\langle x_n, y, z_n\rangle)$ is at least $n$.
\end{proof}

This suggests the following reduction $h$ from $S$ to $A$. The reduction will
have $x_n$ and $z_n$ as advice. Given a string $y$ of length $n$, the
reductions outputs $f(\langle x_n, y, z_n\rangle)$. This reduction is clearly
length-increasing and is length-increasing on every string from $S$.
Thus we have the following lemma.

\begin{lemma}
Consider the above reduction $h$ from $S$ to $A$, for all $y \in S$, $|h(y)| >
|y|$.
\end{lemma}

Now we show how to obtain a length-increasing reduction on all strings. 
We make the following crucial observation.

\begin{observation}
For all but finitely many $n$, there is a string $y_n$ of length $n$ such that
$y_n \notin S$ and $|f(\langle x_n, y_n, z_n\rangle )| > n$.
\end{observation}

\begin{proof}
Suppose not. This means that for infinitely many $n$, for every $y$ from
$\overline{S} \cap \Sigma^n$, the length of $f(\langle x_n, y, z_n\rangle)$ is
less than $n$. Now consider the following algorithm that solves $S$.
Given a string $y$ of length $n$, compute $f(\langle x_n, y, z_n\rangle)$. If
the length of $f(\langle x_n, y, z_n\rangle) > n$, then accept $y$ else reject
$y$.

The above algorithm can be implemented in $\ppoly$ given $x_n$ and $z_n$ as
advice. If $y \in S$, then we know that that the length of $f(\langle x_n, y,
z_n\rangle)$ is bigger than $n$, and so the above algorithm accepts. If $y
\notin S$, then by our assumption, the length of $f(\langle x_n, y,
z_n\rangle)$ is at most $n$. In this case the algorithm rejects $y$. This
shows that $S$ is in $\io\ppoly$ which in turn implies that $H'$ is in
$\io\ppoly$.  This is a contradiction.
\end{proof}

Now we are ready to describe our length increasing reduction from $S$ to $A$.
At length $n$, this reduction will have $x_n$, $y_n$ and $z_n$ as advice.
Given a string $y$ of length $n$, the reduction outputs $f(\langle x_n, y,
z_n\rangle)$ if the length of $f(\langle x_n, y, z_n\rangle)$ is more than $n$.
Else, the reduction outputs $f(\langle x_n, y_n, z_n\rangle)$.

Since $H(x_n) \neq H(z_n)$, $y \in S$ if and only if $f(\langle x_n, y,
z_n\rangle) \in A$. Thus the reduction is correct when it outputs $f(\langle
x_n, y, z_n\rangle)$. The reduction outputs $f(\langle x_n, y_n, z_n\rangle)$
only when the length of $f(\langle x_n, y, z_n\rangle)$ is at most $n$. We
know that in this case $y \notin S$. Since $y_n \notin S$, $f(\langle x_n,
y_n, z_n) \notin A$. 

Thus we have a $\ppoly$-computable, length-increasing from $S$ to $A$.
Thus all $\CoNP$-complete languages are complete via $\ppoly$,
length-increasing reductions. This immediately implies that all $\NP$-complete
languages are complete via $\ppoly$-computable, length-increasing reductions.
\end{proof}

%% file: tm.tex
\section{Separation of Completeness Notions}\label{se:tm}

In this section we consider the question whether the Turing
completeness differs from many-one completeness for $\NP$
under two plausible complexity-theoretic hypotheses:
\begin{enumerate}[(1)]
\item There exists a $2^{n^\epsilon}$-secure one-way permutation.
\item $\NEEE\cap\coNEEE \not\subseteq \EEElog$.
\end{enumerate}
It turns out that the first hypothesis implies that every many-one complete language
for $\NP$ is complete under a particular kind of length-increasing reduction,
while the second hypothesis provides us with a specific Turing complete language that
is not complete under the same kind of length-increasing reduction.
Therefore, the two hypotheses together separate the notions of many-one and
Turing completeness for $\NP$ as stated in the following theorem.

\begin{theorem}\label{TMthm}
If both of the above hypotheses are true, there is is a language that is
polynomial-time Turing complete for $\NP$ but not polynomial-time
many-one complete for $\NP$.
\end{theorem}

Theorem~\ref{TMthm} is immediate from Lemma~\ref{thm:one-way-qp-li}
and Lemma~\ref{lm:3_5} below.

\begin{lemma}\label{thm:one-way-qp-li}
Suppose $2^{n^{\epsilon}}$-secure one-way permutations exist.
Then for every $\NP$-complete language $A$ and every $B\in\NP$, there is
a quasipolynomial-time computable, polynomial-bounded,
length-increasing reduction  reduction $f$ from $B$ to $A$.
\end{lemma}

A function $f$ is {\em polynomial-bounded} if there is a polynomial $p$ such
that the length of $f(x)$ is at most $p(|x|)$ for every $x$.

\begin{lemma}\label{lm:3_5}
If $\NEEE\cap\coNEEE \nsubseteq \EEElog$, then there is a
polynomial-time Turing complete set for $\NP$ that is not many-one
complete via quasipolynomial-time computable, polynomial-bounded,
length-increasing reductions.
\end{lemma}

\ignore{
To show Lemma~\ref{thm:one-way-qp-li} we use the following result of
Agrawal \cite{Agr02}.
Let $\gamma>0$ and let $S\subseteq \{0,1\}^n$.
A function $g$ is $\gamma$-{\em sparsely many}-{\em one} on $S$
if for every $x\in S$,
\[|g^{-1}(g(x))\cap \{0,1\}^n|\leq \frac{2^n}{2^{n^\gamma}}.\]
\begin{lemma}[Agrawal \cite{Agr02}]\label{sparselemma}
Suppose $2^{n^{\epsilon}}$-secure one-way permutations exists. For
every NP-complete language $L$ for every set $S$ in $\NP$, there is a
reduction from $f$ from $S$ to $L$ that is $\frac{\epsilon}{2}$-sparsely
many-one on $\{0,1\}^n$ for all $n\in\N$.
\end{lemma}

\newcommand{\bd}[1]{\langle #1 \rangle}

\begin{proof}[Proof of Theorem \ref{thm:one-way-qp-li}]
Let $\delta = \epsilon/3$.
Consider the following set
\[S = \myset{\bd{x,y}}{x\in B\text{ and }|y| = n^{1/\delta}}.\]
For $|x|=n$, let $m = |\bd{x,y}|$.
By Lemma~\ref{sparselemma}, there is a reduction $f$ from $S$ to $A$ that
is $\frac{\epsilon}{2}$-sparsely many-one on $\Sigma^m$ for every $m$.
For any $x$ of length $n$, let $S_x = \{y \in \Sigma^{n^{1/\delta}}~|~ |f(\langle x, y\rangle)| \leq n\}$.
It follows that the size of $S_x$ is at most $\frac{2^{n+m+1}}{2^{m^{\epsilon/2}}}$. Thus for at least $1-
2^{-2n^{3/2}}$ fraction of strings $y$ from $\{0,1\}^{n^{1/\delta}}$,
$|f(\bd{x,y})|>n$. If we randomly pick a $y \in \{0,1\}^{n^{1/\delta}}$,
then with very high probability $|f(\pair{x,y})|>n$.

We can derandomize the above process.
If $2^{n^{\epsilon}}$-secure one-way permutations exist, then $\EXP$ does not have $2^{n^{\epsilon}}$-size circuits.
Thus by Theorem~\ref{lm:pseudorandom}, there exists a pseudorandom generator family $\{G_n\}$ that map strings of length
$\log^k n$ to strings of length $n$. These pseudorandom generators are
computable in $O(2^{\log^d n})$ time for some constant $d$.

Now the length-increasing reduction from $B$ works as follows.
Given $x$ as input of length $n$, let $t = n^{1/\delta}$.
Cycle through all seeds $s$ of length $\log^k t$ till the length of $f(\bd{x,G_t(s)})$
is bigger than $n$. Output $f(\bd{x,G_t(s)})$. Since $G_t$ is a
pseudorandom generator, it follows that
there exists at least one $s$ for which the length $f(\bd{x,G_t(s)})$ is bigger than $n$.

Clearly, the above reduction can be computed in time quasipolynomial
in $n$.  The output of the reduction is bounded by the length of
$f(\bd{x,y})$, where $|y|=n^{1/\delta}$.  Since $f$ is polynomial-time
computable, the reduction is polynomial-bounded.
\end{proof}
}

\newcommand{\T}{\mathcal{T}}
The proof of Lemma~\ref{thm:one-way-qp-li} will appear in the full paper.
The remainder of this section is devoted to proving Lemma \ref{lm:3_5}.
It is well known that any set $A$ over $\Sigma^*$ can be encoded as a
tally set $T_A$ such that $A$ is worst-case hard if and only if $T_A$
is worst-case hard. For our purposes, we need an average-case version
of the this equivalence.  Below we describe particular encoding of
languages using tally sets that is helpful for us and prove the
average-case equivalence.

Let $t_0 = 2$, $t_{i+1} = t_i^2$ for all $i\in\N$.
Let $\T =\myset{0^{t_i}}{i\in\N}$.
For each $l\in\N$, let $\T_l = \myset{0^{t_i}}{2^l-1\leq i\leq
2^{l+1}-2}$.
Observe that $\T=\bigcup_{l=0}^\infty \T_l$.
Given a set $A \subseteq \{0, 1\}^*$, let
$T_A =\mysetl{0^{2^{2^{r_x}}}}{ x\in A},$
where $r_x$ is the rank index of $x$ in the standard enumeration of
$\strings$.
It is easy to verify that for all $l\in\N$ and every $x$,
\begin{eqnarray}
x\in A \cap \{0,1\}^l & \iff  & 0^{t_{r_x}}\in T_A\cap \T_l.\label{eq:one}
\end{eqnarray}


\begin{lemma}\label{lm:3_3}
Let $A$ and $T_A$ be as above.  Suppose there is a quasipolynomial
time algorithm $\sA$ such that for every $l$, on an $\epsilon$ fraction
of strings from $\T_l$, this algorithm correctly decides the
membership in $T_A$, and on the rest of the strings the algorithm
outputs ``I do not know''.  There is a $2^{2^{2^{k(l+1)}}}$-time
algorithm $\sA'$ for some constant $k$ that takes one bit of advice
and correctly decides the membership in $A$ on
$\frac{1}{2}+\frac{\epsilon}{2}$ fraction of the strings at every
length $l$.
\end{lemma}


We know several results that establish worst-case to average-case
connections for classes such as $\EXP$ and
$\PSPACE$~\cite{Yao82,BFNW93,Imp95,ImpWig97,SudanTrevisanVadhan01}.
The following lemma establishes a similar connection for triple
exponential time classes, and can be proved using known techniques.

\begin{lemma}\label{lm:3_4}
If $\NEEE\cap \coNEEE \not\subseteq \EEElog$, then there is language $L$ in $\NEEE\cap\co\NEEE$
such that no $\EEElog$ algorithm can decide $L$, at infinitely many
lengths $n$, on more than $\frac{1}{2}+\frac{1}{n}$ fraction of strings
from $\{0,1\}^n$.
\end{lemma}

Now we are ready to prove Lemma~\ref{lm:3_5}.

\begin{proof}[Proof of Lemma \ref{lm:3_5}]
  By Lemma \ref{lm:3_4}, there is a language $L\in (\NEEE\cap \coNEEE)
  - \EEElog$ such that no $\EEElog$ algorithm can decide $L$
  correctly on more than a $\frac{1}{2}+\frac{1}{n}$ fraction of the
  inputs for infinitely many lengths $n$.

Without loss of generality, we can assume that
$L\in\NTIME(2^{2^{2^{n}}})\cap \co\NTIME(2^{2^{2^{n}}})$
Let
\[T_L=\mysetl{0^{2^{2^{r_x}}}}{ x\in L}.\]

Clearly, $T_L \in \NP\cap\co\NP$.

Define $\tau:\N\rightarrow\N$ such that $\tau(n) =\max\myset{i}{t_i\leq n}$.
Now we will define our Turing complete language.
Let
\[\SAT_0 =\myset{0x}{0^{t_{\tau(|x|)}}\notin T_L \text{ and }x\in\SAT},\]
\[\SAT_1 =\myset{1x}{0^{t_{\tau(|x|)}}\in T_L \text{ and }x\in\SAT}.\]
Let $A = \SAT_0\cup \SAT_1$. Since $L$ is in $\NP \cap \CoNP$, $A$ is in
$\NP$. The following is a Turing reduction from $\SAT$ to $A$: Given
a formula $x$, ask queries $0x$ and $1x$, and accept if and only if at
least one them is in $A$.
Thus $A$ is polynomial-time $2$-$\mathrm{tt}$ complete for $\NP$.

Suppose $A$ is complete via length-increasing, polynomial-bounded,
quasipolynomial-time reductions. Then there is such a reduction $f$ from $\{0\}^*$ to $A$.
There is a constant $d$ such that $f$ is $n^d$-bounded and runs in quasipolynomial time.

The following observation is easy to prove.
\begin{observation}\label{easyobs}
Let $y\in\strings$ and $b\in\{0,1\}$ be such that $f(0^{t_i}) = by$.
Then $0^{t_{\tau(|y|)}}\in T_L$ if and only if $b = 1$.
\end{observation}

Fix a length $l$. We will describe a quasipolynomial-time
algorithm that will decide the membership in $T_L$ on at least $\frac{1}{\log d}$
fraction of strings from $\T_l$, and says ``I do not know'' on
other strings. By the Lemma \ref{lm:3_3}, this implies
that there is $\EEE/1$ algorithm that decides $L$ on more than
$\frac{1}{2} +\frac{1}{2\log d}$ fraction of strings from $\{0,1\}^l$. This contradicts
the hardness of $L$ and completes the proof.

Let $s = 2^l -1$ and $r = 2^{l+1}-2$.
Recall that $\T_l = \myset{0^{t_i}}{s \leq i \leq r}$. Divide $\T_l$ in
sets $T_0, T_2, \cdots T_r$ where $T_k = \myset{0^{t_i}}{s+k\log d \leq
r+(k+1)\log d}$. This gives at least $\frac{2^l}{\log d}$ sets.
Consider the following algorithm that decides $T_L$ on strings from
$\T_l$: Let $0^{t_j}$ be the input. Say, it lies in the set $T_k$.
Compute $f(0^{t_{s+k\log d}}) = by$.
If $t_{\tau(|y|)} \neq t_j$, then output ``I do not know''. Otherwise, accept $0^{t_j}$ if and only
if $b = 1$. By Observation \ref{easyobs} this algorithm never errs.
Since $f$ is computable in quasipolynomial time, this algorithm runs in quasipolynomial time.
Finally, observe that $t_{\tau(|y|)}$ lies
between $t_{s+k\log d}$ and $t_{s+(k+1)\log d}$. Thus for every $k$,
$0 \leq k \leq r$, there is at least one string from from $T_k$ on which
the above algorithm correctly decides $T_L$. Thus the above algorithm
correctly decides $T_L$ on at least $\frac{1}{\log d}$ fraction of strings from
$\T_l$, and never errs.
\end{proof}